\documentclass[a4paper, UKenglish,cleveref, autoref, thm-restate]{lipics-v2021}
\pdfoutput=1 %uncomment to ensure pdflatex processing (mandatatory e.g. to submit to arXiv)
\hideLIPIcs  %uncomment to remove references to LIPIcs series (logo, DOI, ...), e.g. when preparing a pre-final version to be uploaded to arXiv or another public repository

\author{Geri Gokaj}{Karlsruhe Institute of Technology, Karlsruhe, Germany}{geri.gokaj@kit.edu}{https://orcid.org/0009-0002-7500-6848}{}%TODO mandatory, please use full name; only 1 author per \author macro; first two parameters are mandatory, other parameters can be empty. Please provide at least the name of the affiliation and the country. The full address is optional. Use additional curly braces to indicate the correct name splitting when the last name consists of multiple name parts.
\title{Completeness Theorems for $k$-SUM and Geometric Friends: Deciding Fragments of Linear Integer Arithmetic} %TODO Please add
\titlerunning{Completeness Theorems for $k$-SUM and Geometric Friends}
\author{Marvin Künnemann}{Karlsruhe Institute of Technology, Karlsruhe, Germany}{marvin.kuennemann@kit.edu}{}{}
\authorrunning{G. Gokaj and M. Künnemann} %TODO mandatory. First: Use abbreviated first/middle names. Second (only in severe cases): Use first author plus 'et al.'
\Copyright{Geri Gokaj and Marvin Künnemann} %TODO mandatory, please use full first names. LIPIcs license is "CC-BY";  http://creativecommons.org/licenses/by/3.0/
 %TODO mandatory: Please choose ACM 2012 classifications from https://dl.acm.org/ccs/ccs_flat.cfm 

 %TODO mandatory; please add comma-separated list of keywords

\category{} %optional, e.g. invited paper

\relatedversion{} %optional, e.g. full version hosted on arXiv, HAL, or other respository/website
\funding{This research is supported by the Deutsche Forschungsgemeinschaft (DFG, German Research Foundation) - 462679611.}%optional, to capture a funding statement, which applies to all authors. Please enter author specific funding statements as fifth argument of the \author macro.

\acknowledgements{The authors like to thank the reviewers for constructive feedback as well as Karl Bringmann and Nick Fischer for helpful discussions. }%optional

\nolinenumbers %uncomment to disable line numbering

\EventEditors{Raghu Meka}
\EventNoEds{1}
\EventLongTitle{16th Innovations in Theoretical Computer Science Conference (ITCS 2025)}
\EventShortTitle{ITCS 2025}
\EventAcronym{ITCS}
\EventYear{2025}
\EventDate{January 7--10, 2025}
\EventLocation{Columbia University, New York, NY, USA}
\EventLogo{}
\SeriesVolume{325}
\ArticleNo{55}
\usepackage{amsthm,thmtools,amsfonts,amsmath}
\usepackage{mathtools}
\usepackage{xcolor, graphicx}
\usepackage{algorithm}
\usepackage{algpseudocode}

\newcommand{\FOPZ}{\mathsf{FOP}_{\mathbb{Z}}}
\newcommand{\PTO}{\mathsf{PTO}}
\newcommand{\tOh}{\tilde{O}}

\usepackage{hyperref}

\date{}

\begin{document}

\maketitle

%TODO mandatory: add short abstract of the document
\begin{abstract}
In the last three decades, the $k$-SUM hypothesis has emerged as a satisfying explanation of long-standing time barriers for a variety of algorithmic problems. Yet to this day, the literature knows of only few proven consequences of a refutation of this hypothesis. Taking a descriptive complexity viewpoint, we ask: What is the largest logically defined class of problems \emph{captured} by the $k$-SUM problem?

	To this end, we introduce a class $\mathsf{FOP}_{\mathbb{Z}}$ of problems corresponding to deciding sentences in Presburger arithmetic/linear integer arithmetic over finite subsets of integers.
	We establish two large fragments for which the $k$-SUM problem is complete under fine-grained reductions:
	\begin{enumerate}
		\item The $k$-SUM problem is complete for deciding the sentences with $k$ existential quantifiers.
		\item The $3$-SUM problem is complete for all $3$-quantifier sentences of $\mathsf{FOP}_{\mathbb{Z}}$ expressible using at most $3$ linear inequalities.
	\end{enumerate}
	Specifically, a faster-than-$n^{\lceil k/2 \rceil \pm o(1)}$ algorithm for $k$-SUM (or faster-than-$n^{2 \pm o(1)}$ algorithm for $3$-SUM, respectively) directly translate to polynomial speedups of a general algorithm for \emph{all} sentences in the respective fragment.
	
	Observing a barrier for proving completeness of $3$-SUM for the entire class $\FOPZ$, we turn to the question which other -- seemingly more general -- problems are complete for $\FOPZ$. In this direction, we establish $\FOPZ$-completeness of the \emph{problem pair} of Pareto Sum Verification and Hausdorff Distance under $n$ Translations under the $L_\infty$/$L_1$ norm in $\mathbb{Z}^d$. In particular, our results invite to investigate Pareto Sum Verification as a high-dimensional generalization of 3-SUM.

\end{abstract}

\section{Introduction}
Consider a basic question in complexity theory: How can we determine for which problems an essentially quadratic-time algorithm is best possible? If a given problem $A$ admits an algorithm running in $n^{2+o(1)}$ time, and it is known that $A$ cannot be solved in time $O(n^{2-\epsilon})$ for any $\epsilon>0$, then clearly the $n^{2+o(1)}$ algorithm has \emph{optimal} runtime, up to subpolynomial factors.
This question can be asked more generally for any $k\geq 1$ and time $n^{k \pm o(1)}$.
To this day, the theoretical computer science community is far from able to resolve this question unconditionally. However, a surge of results over recent years uses
conditional lower bounds based on plausible hardness assumptions to shed some light on why some problems seemingly cannot be solved in time $O(n^{k-\epsilon})$ for any $\epsilon >0$.
Most notably, reductions from $k$-OV, $k$-SUM and the weighted $k$-clique problem have been used to establish $n^{k-o(1)}$-time conditional lower bounds, often matching known algorithms; see~\cite{williams2018some} for a detailed survey.

In this context, the $3$-SUM hypothesis is arguably the first -- and particularly central -- 
hardness assumption for conditional lower bounds. Initially introduced to explain various quadratic-time barriers
observed in computational geometry~\cite{DBLP:journals/comgeo/GajentaanO95}, it has since been used to 
show quadratic-time hardness for a wealth of problems from various 
fields~\cite{DBLP:journals/siamcomp/WilliamsW13, DBLP:conf/stoc/Patrascu10,DBLP:conf/focs/AbboudW14, DBLP:conf/soda/KopelowitzPP16, DBLP:conf/stoc/0001GS20, DBLP:conf/stoc/AbboudBKZ22, DBLP:conf/stoc/ChanWX22}.
Its generalization, the $k$-SUM\footnote{The $k$-SUM problem asks, given sets $A_1,\dots ,A_k$ of $n$ numbers, whether there exist $a_1 \in A_1 ,\dots ,a_k \in A_k$ such that $\sum_{i=1}^{k}a_i=0$.
The $k$-SUM hypothesis states that for no $\epsilon>0$ there exists a $O(n^{\lceil k/2 \rceil-\epsilon})$ time algorithm that solves $k$-SUM.} hypothesis, has led to further conditional lower bounds beyond the quadratic-time
regime~\cite{DBLP:journals/siamcomp/Erickson99, DBLP:conf/icalp/AbboudL13, DBLP:conf/focs/AbboudBBK17, DBLP:conf/nips/AbboudBBK20, DBLP:conf/focs/Kunnemann22}. For a more comprehensive overview, we refer to~\cite{williams2018some}.

The centrality of the $3$-SUM hypothesis for understanding quadratic-time barriers begs an interesting question: Does $3$-SUM fully capture quadratic-time solvability, in the sense that it is hard for the entire class $\mathsf{DTIME}(n^2)$? Alas, Bloch, Buss, and Goldsmith~\cite{bloch1994hard} give evidence that we are unlikely to prove this: If $3$-SUM is hard for $\mathsf{DTIME}(n^2)$ under quasilinear reductions, then $\mathsf{P} \ne \mathsf{NP}$.
%~\cite{DBLP:conf/cccg/Barrera96,DBLP:journals/ijcga/BarequetH01,DBLP:journals/comgeo/GajentaanO95} \geri{mehr}.
Thus, to understand precisely the role of $3$-SUM to understand quadratic-time computation, the more reasonable question to ask is: 

\begin{center}
\em What is the largest class $\mathcal{C}$ of problems such that $3$-SUM is $\mathcal{C}$-hard?\footnote{Note that there are different reasonable notions of reductions to consider. Rather than the quasilinear reductions used by Bloch et al., we will consider the currently more commonly used notion of fine-grained reductions; see Section~\ref{sec:contributions} for details on the notion of completeness that we will use.}
\end{center}

Finding a large class $\mathcal{C}$ for which $3$-SUM is hard may be seen as giving evidence for the $3$-SUM hypothesis. Furthermore, such a result may clarify the true expressive power of the $3$-SUM hypothesis, much like the NP-completeness of 3-SAT highlights its central role for polynomial intractability.

\subsection{Our approach}
We approach our central question from a descriptive complexity perspective. This line of research
has been initiated by Gao et al.~\cite{GaoIKW19}, who establish the sparse OV problem as complete for the class of model checking first-order properties.
One can interpret this result as showing that the OV problem expresses relational database queries in the sense that a truly subquadratic algorithm for OV would improve the fine-grained data 
complexity of such queries (see~\cite{GaoIKW19} for details). Related works further delineate the fine-grained hardness of model checking first-order properties and related 
problem classes~\cite{DBLP:conf/csl/Williams14, DBLP:conf/coco/BringmannFK19, DBLP:conf/approx/BringmannCFK21, DBLP:journals/algorithmica/AnGIJKN22, DBLP:conf/icalp/BringmannCFK22,FischerKR24}, see Section~\ref{sec:relatedWork} for more discussion. %uniform version of the circuit complexity class $\mathsf{AC}^0$. 

Towards continuing the line of research on fine-grained completeness theorems,
we introduce a class of problems corresponding to deciding formulas in linear integer arithmetic over finite sets of integers. Specifically, consider the vectors $$x_1 =(x_1[1],\dots, x_1[d_1]) , \dots, x_k=(x_k[1],\dots, x_k[d_k])$$ as quantified variables, and let $t_1,\dots ,t_l$ be free variables.
Moreover, let 
\[X:=\{x_1[1],\dots, x_1[d_1], \dots x_k[1], \dots, x_k[d_k], t_1,\dots ,t_l\},\]
and let $\psi$ be a quantifier-free linear arithmetic formula over variables in $X$.
We consider the model-checking problem for formulas $\phi$ in the prenex normal form 
\[ \phi:=Q_1 x_1 \dots Q_k x_k: \psi, \]
where the quantifiers $Q_1, \dots, Q_k \in \{ \exists, \forall \}$ are arbitrary. Formally, for such a $\phi$, we define the model checking problem $\mathsf{FOP}_{\mathbb{Z}}(\phi)$ as follows\footnote{Below, we use the notation $\psi[(t_1,\dots, t_l) \backslash (\hat{t_1},\dots, \hat{t_l})]$ to denote the substitution of the variables $t_1,\dots,t_l$ by $\hat{t_1},\dots,\hat{t_l}$ respectively.}
\begin{align}
	\mathsf{FOP}_{\mathbb{Z}}(\phi):  & \label{model_check}\\
\textbf{Input:}& \text{ Finite sets }A_1 \subseteq \mathbb{Z}^{d_1} ,\dots, A_k \subseteq  \mathbb{Z}^{d_k}\text{ and }\hat{t_1},\dots, \hat{t_l} \in \mathbb{Z}. \notag \\
\textbf{Problem:} &\text{ Does } Q_1 x_1 \in A_1 \dots Q_k x_k \in A_k : \psi[(t_1,\dots, t_l) \backslash (\hat{t_1},\dots, \hat{t_l})] \text{ hold?} \notag
\end{align}
We let $n := \max_i\{|A_i|\}$ denote the input size and will assume throughout the paper that all input numbers (i.e., the coordinates of the vectors in $A_1,\dots, A_k$ and the values $\hat{t_1},\dots, \hat{t_l}$) are chosen from a polynomially sized universe, i.e., $\{-U, \dots, U\}$ with $U\le n^{c}$ for some~$c$.
Let $\mathsf{FOP}_{\mathbb{Z}}$ be the union of all $\mathsf{FOP}_{\mathbb{Z}}(\phi)$ problems, where $\phi$ has at least 3 quantifiers.\footnote{It is not too difficult to see that all formulas with 2 quantifiers can be model-checked in near-linear time; see the full version for details.}
Besides $3$-SUM, a variety of interesting problems is contained in $\mathsf{FOP}_{\mathbb{Z}}$; we discuss a few notable examples below and in the full version.

Frequently, we will distinguish formulas in $\mathsf{FOP}_{\mathbb{Z}}$ using their quantifier structure; e.g., $\mathsf{FOP}_{\mathbb{Z}}(\exists \exists \forall)$ describes
the class of model checking problems $\mathsf{FOP}_{\mathbb{Z}} (\phi)$ where in $\phi$ we have $Q_1=Q_2=\exists $ and $Q_3=\forall.$ 
Furthermore, we let $\mathsf{FOP}_{\mathbb{Z}}^k$ be the union of all $\mathsf{FOP}_{\mathbb{Z}}(\phi)$ problems, where $\phi$ consists of precisely $k$ quantifiers, regardless of their quantifier structure.
For a quantifier $Q \in \{\exists,\forall\}$, we write $Q^k$ for the repetition $\underbrace{Q \dots Q}_{k \text{ times}}$.
Finally, we remark that a small subset of $\mathsf{FOP}_{\mathbb{Z}}$ has already been studied by An et al. \cite{DBLP:journals/algorithmica/AnGIJKN22}, for a discussion see Section \ref{sec:relatedWork}.
%Finally, let $\mathsf{FOP}_{\mathbb{Z}}$ be the class of all model checking problems for valid $\phi$. Let us consider an application of 
%$\exists^3 \mathsf{FOP}_{\mathbb{Z}}$. \geri{The counting versions of the problems are naturally defined, take the model checking problem and ask for the number of solutions.}
%\geri{Then just make the classes analogously.}
%In the first part of the paper we focus on the class $\exists^k \mathsf{FOP}_{\mathbb{Z}}$ for $k\geq 3$.
%These classes of problems are able to model a variety of tasks. For instance, assuming the 3-uniform hyperclique hypothesis,

\subsection{Our Contributions}\label{sec:contributions}
%In the following, we go through our contributions at a high level. 
%For more technical results, and a detailed description of our results for each class with precisely $3$ quantifiers, see the Technical Overview Section. 

We seek to determine completeness results for the class $\mathsf{FOP}_{\mathbb{Z}}$.
In particular: What are the largest fragments of this class for which 3-SUM (or more generally, $k$-SUM) is complete? Is there a problem that is complete for the entire class?

Intuitively, we say that a $T_A(n)$-time solvable problem $A$ is \emph{(fine-grained) complete} for a $T_{\mathcal{C}}(n)$-time solvable class of problems $\mathcal{C}$, if the existence of an $O(T_A(n)^{1-\epsilon})$-time algorithm for~$A$ with $\epsilon>0$ implies that for \emph{all} problems $C$ in $\mathcal{C}$ there exists $\delta > 0$ such that $C$ can be solved in time $O(T_{\mathcal{C}}(n)^{1-\delta})$. We extend this notion to completeness of a \emph{family} of problems, since strictly speaking, any (geometric) problem over $\mathbb{Z}^d$ expressible in linear integer arithmetic corresponds to a family of formulas $\mathsf{FOP}_{\mathbb{Z}}$ (one for each $d\in \mathbb{N}$).
Formally, consider a family of problems $\mathcal{P}$ with an associated time bound $T_{\mathcal{P}}(n)$ and a class of problems~$\mathcal{C}$ with an associated time bound $T_{\mathcal{C}}(n)$; usually $T_{\mathcal{P}}(n), T_{\mathcal{C}}(n)$ denote the running time of the fastest known algorithm solving all problems in $\mathcal{P}$ or $\mathcal{C}$, respectively (often, we omit these time bounds, as they are clear from context).\footnote{Here, we use \emph{family} and \emph{class} as a purely semantic and intuitive distinction: A family consists of a small set of similar problems, and a class consists of a large and diverse variety of problems.}
We say that $\mathcal{P}$ is \emph{(fine-grained) complete} for $\mathcal{C}$, if
\begin{enumerate}
\item the family $\mathcal{P}$ is a subset of the class $\mathcal{C}$, and
\item if for all problems $P$ in $\mathcal{P}$ there exists $\epsilon>0$ such that $P$ can be solved in time $O(T_{\mathcal{P}}(n)^{1-\epsilon})$, then for all problems $C$ in $\mathcal{C}$ there exists some $\delta>0$ such that we can solve $C$ in time $O(T_{\mathcal{C}}(n)^{1-\delta})$.
\end{enumerate}
That is, a polynomial-factor improvement for solving the problems in $\mathcal{P}$ would lead
to a polynomial-factor improvement in solving \emph{all} problems in $\mathcal{C}$.
If a singleton family $\mathcal{P} = \{P\}$ is fine-grained complete for $\mathcal{C}$, then we also say that $P$ is fine-grained complete for $\mathcal{C}$.
We work with standard hypotheses and problems encountered in fine-grained complexity; for detailed definitions of these, we refer to the full version of this article.

\subsubsection{$k$-SUM is complete for the existential fragment of $\mathsf{FOP}_{\mathbb{Z}}$}
Consider first the existential fragment of $\FOPZ$, i.e., formulas exhibiting only existential quantifiers.
Any $\mathsf{FOP}_{\mathbb{Z}}$ formula with $k$ existential quantifiers can be decided using a standard meet-in-the-middle approach, augmented by orthogonal range search, in time $\Tilde{O}(n^ {\lceil k/2 \rceil})$\footnote{We use the notation $\Tilde{O}(T):=T \log^{O(1)}(T)$ to hide polylogarithmic factors.}, see the full version of the paper for details. Since $k$-SUM is a member of $\FOPZ(\exists^k)$, this running time is optimal up to subpolynomial factors, assuming the $k$-SUM Hypothesis.
As our first contribution, we provide a converse reduction. Specifically, we show that a polynomially improved $k$-SUM algorithm would give a polynomially improved algorithm for solving the entire class. In our language, we show that $k$-SUM is fine-grained complete for formulas of $\mathsf{FOP}_{\mathbb{Z}}$ with $k$ existential quantifiers.

\begin{restatable}[$k$-SUM is $\mathsf{FOP}_{\mathbb{Z}}(\exists^k)$-complete]{theorem}{threesumcompl}
	Let $k\geq 3$ and assume that $k$-SUM can be solved in time $T_{k\mathrm{SUM}}(n)$. For any problem $P$ in $\mathsf{FOP}_{\mathbb{Z}}(\exists^k )$, there exists some $c$ such that $P$ can be solved in time $O(T_{k\mathrm{SUM}}(n) \log^c n)$. 
\label{existential_complete}
\end{restatable}

Thus, if there are $k\geq 3$ and $\epsilon>0$ such that we can solve $k$-SUM in time $O(n^{\lceil k/2 \rceil-\epsilon})$, then we can solve all problems in $\mathsf{FOP}_{\mathbb{Z}}(\exists^k )$ in time  $O(n^{\lceil k/2 \rceil-\epsilon'})$ for any $0 < \epsilon' < \epsilon$.
By a simple negation argument, we conclude that $k$-SUM is also complete for the class of problems $\mathsf{FOP}_{\mathbb{Z}}(\forall^k)$.

The above theorem generalizes and unifies previous reductions from problems expressible as $\FOPZ(\exists^k)$ formulas to 3-SUM, using different proof ideas: Jafargholi and Viola~\cite[Lemma 4]{DBLP:journals/algorithmica/JafargholiV16} give a simple randomized linear-time reduction from triangle detection in sparse graphs to 3-SUM, and a derandomization via certain combinatorial designs.
Dudek, Gawrychowski, and Starikovskaya~\cite{DBLP:conf/stoc/0001GS20} study the family of 3-linear degeneracy testing (3-LDT), which constitutes a large and interesting subset of $\FOPZ(\exists \exists \exists)$: This family includes, for any $\alpha_1,\alpha_2,\alpha_3, t\in \mathbb{Z}$, the \emph{3-partite} formula $\exists a_1 \in A_1 \exists a_2\in A_2 \exists a_3\in A_3: \alpha_1 a_1 + \alpha_2 a_2 + \alpha_3 a_3 = t$ and the \emph{1-partite} formula $\exists \alpha_1, \alpha_2, \alpha_3 \in A: \alpha_1 a_1 + \alpha_2 a_2 + \alpha_3 a_3 = t \wedge a_1\ne a_2 \wedge a_2\ne a_3 \wedge a_1\ne a_3$.
The authors show that each such formula is either trivial or subquadratic \emph{equivalent} to $3$-SUM.
For 3-partite formulas, a reduction to $3$-SUM is essentially straightforward.
For 1-partite formulas, Dudek et al.~\cite{DBLP:conf/stoc/0001GS20} use color coding.\footnote{We remark that the reverse direction, i.e., $3$-SUM-hardness of non-trivial formulas, is technically much more involved and can be regarded as the main technical contribution of~\cite{DBLP:conf/stoc/0001GS20}.} %At this point we would like to differentiate between the aim of our paper and the work of Dudek et al. .

As further examples for reductions from $\FOPZ$ problems to $k$-SUM, we highlight a reduction from Vector $k$-SUM to $k$-SUM~\cite{DBLP:journals/corr/AbboudLW13} as well as a reduction from $(\min,+)$-convolution to $3$-SUM (see~\cite{BackursIS17,DBLP:journals/talg/CyganMWW19}) based on a well-known bit-level trick due to Vassilevska Williams and Williams~\cite{DBLP:journals/siamcomp/WilliamsW13}, which allows us to reduce inequalities to equalities. 

Perhaps surprisingly in light of its generality and applicability, Theorem~\ref{existential_complete} is obtained via a very simple, deterministic reduction that combines the tricks from~\cite{DBLP:journals/corr/AbboudLW13,DBLP:journals/siamcomp/WilliamsW13}. This generality comes at the cost of polylogarithmic factors (which we do not optimize), which depend on the number of inequalities occurring in the considered formula; for the details see Section \ref{sec:existentialpa} and the full version of the paper.

\subsubsection{Completeness for counting witnesses}
We provide a certain extension of the above completeness result to the problem class of \emph{counting} witnesses to existential $\FOPZ$ formulas\footnote{A witness for a $\FOPZ(\exists^k)$ formula $\exists a_1 \in A_1 \dots \exists a_k \in A_k: \varphi$ with $\hat{t_1}, \dots, \hat{t_l} \in \mathbb{Z}$ is a tuple $(a_1,\dots, a_k)\in A_1\times \cdots \times A_k$ that satisfies the formula $\varphi[(t_1,\dots,t_l) \backslash (\hat{t_1},\dots,\hat{t_l})]$.}. Counting witnesses is an important task particularly in database applications (usually referred to as model counting). Furthermore, we will make use of witness counting to \emph{decide} certain quantified formulas in subsequent results detailed below. In Section \ref{sec:countwitn}, we will obtain the following result. 

\begin{restatable}{theorem}{countComplete}
	Let $k\geq3$ be odd. If there is $\epsilon > 0$ such that we can count the number of witnesses for $k$-SUM in time $O(n^{\lceil k/2 \rceil -\epsilon})$, then for all problem $P$ in $\FOPZ(\exists^k)$, there is some $\epsilon'> 0$ such that we can count the number of witnesses for $P$ in time $O(n^{\lceil k/2 \rceil -\epsilon'})$.
	\label{thm:counting-witnesses}
\end{restatable}

Leveraging the recent breakthrough by~\cite{DBLP:journals/corr/abs-2303-14572} that 3-SUM is subquadratic equivalent to counting witnesses of 3-SUM, we obtain the corollary that \emph{3-SUM is hard even for counting witnesses of $\FOPZ(\exists^3)$}.

\begin{restatable}{corollary}{threecountcomplete}
	For all problems $P$ in $\mathsf{FOP}_{\mathbb{Z}}(\exists^3)$, there is some $\epsilon_P >0$ such that we can count the number of witnesses for $P$ in randomized time $O(n^{2-\epsilon_P})$ if and only if there is some $\epsilon' > 0$ such that
    $3$-SUM can be solved in randomized time $O(n^{2-\epsilon'})$. 
\label{Count3COMP}
\end{restatable}

\subsubsection{Completeness for general quantifier structures of $\mathsf{FOP}_{\mathbb{Z}}$ }

In light of our first completeness result, one might wonder whether $k$-SUM is complete for deciding all $k$-quantifier formulas in $\mathsf{FOP}_{\mathbb{Z}}$,
regardless of the quantifier structure of the formulas.
Note that for these general quantifier structures, a baseline algorithm with running time $\Tilde{O}(n^{k-1})$ can be achieved with a combination of brute-force and orthogonal range queries; see the full version for details. 

However, by \cite[Theorem 15]{DBLP:journals/algorithmica/AnGIJKN22} there exists a $\FOPZ(\exists^{k-1} \forall)$-formula $\phi$ that cannot be solved in time $O(n^{k-1-\epsilon})$-time unless the 3-uniform hyperclique hypothesis is false (see the discussion in Section~\ref{sec:relatedWork}). Thus, proving that 3-SUM is complete for all 3-quantifier formulas would establish that the 3-uniform hyperclique hypothesis implies the $3$-SUM hypothesis -- this would be a novel tight reduction among important problems/hypotheses in fine-grained complexity theory. For $k\ge 4$, it becomes even more intricate: the conditionally optimal running time of $n^{k-1\pm o(1)}$ for $\FOPZ(\exists^k \forall)$ formulas exceeds the conditionally optimal running time of $n^{\lceil \frac{k}{2} \rceil \pm o(1)}$ for $\FOPZ(\exists^k)$ formulas.  

We are nevertheless able to obtain a completeness result for general quantifier structures: Specifically, we show that if two geometric problems over $\mathbb{Z}^d$ can be solved in time $O(n^{2-\epsilon_d})$ where $\epsilon_d > 0$ for all $d$, then each $k$-quantifier formula in $\FOPZ$ can be decided in time $O(n^{k-1-\epsilon})$ for some $\epsilon > 0$. These problems are (1) a variation of the Hausdorff distance that we call \emph{Hausdorff distance under $n$ Translations} and (2) the Pareto Sum problem; the details are covered in Section \ref{sec:GeneralQuantifier}.

\paragraph*{Hausdorff Distance under $n$ Translations}
Among the most common translation-invariant distance measures for given point sets $B$ and $C$ is the Hausdorff Distance under Translation~\cite{chew1992improvements, DBLP:conf/compgeom/BringmannN21, DBLP:conf/compgeom/Chan23, DBLP:journals/dcg/ChewDEK99, Nusser21, HuttenlocherK90}. 
To define it, we denote the directed Hausdorff distance under the $L_{\infty}$ metric by $\delta_{\overrightarrow{H}}(B,C):= \max_{ b \in B} \min_{c \in C} \|b-c\|_{\infty}$.\footnote{Since we will exclusively consider the \emph{directed} Hausdorff distance under Translation, we will drop ``directed'' throughout the paper.} The Hausdorff distance under translation $\delta_{\overrightarrow{H}}^T(B,C)$ is defined as the minimum Hausdorff distance of $B$ and an arbitrary translation of $C$, i.e.,
\[ \delta_{\overrightarrow{H}}^T(B,C) \coloneqq \min_{\tau \in \mathbb{R}^d} \delta_{\overrightarrow{H}}(B,C+\{\tau\}) = \min_{\tau \in \mathbb{R}^d} \max_{b\in B} \min_{c\in C} \|b-(c+\tau)\|_{\infty}.\]
For $d=2$, Bringmann et al.~\cite{DBLP:conf/compgeom/BringmannN21} were able to show a $(|B||C|)^{1-o(1)}$ time lower bound based on the orthogonal vector hypothesis, and there exists a matching $\Tilde{O}(|B||C|)$ upper bound by Chew et al.~\cite{ChewK92}.

We shall establish that restricting the translation vector to be among a set of $m$ candidate vectors yields a central problem in $\FOPZ$. Specifically, we define the Hausdorff distance under Translation in $A$, denoted as $\delta_{\overrightarrow{H}}^{T(A)}(B,C)$, by
 \[\delta_{\overrightarrow{H}}^{T(A)}(B,C) \coloneqq \min_{\tau \in A} \delta_{\overrightarrow{H}}(B,C+\{\tau\}) = \min_{\tau \in A} \max_{b\in B} \min_{c\in C} \|b-(c+\tau)\|_{\infty}.\]
Correspondingly, we define the problem \emph{Hausdorff distance under $m$ Translations} as: Given $A, B, C\subseteq \mathbb{Z}^d$ with $|A|\le m$, $|B|,|C| \le n$ and a distance value $\gamma \in \mathbb{N}$, determine whether $\delta_{\overrightarrow{H}}^{T(A)}(B,C) \le \gamma$.
Note that this can be rewritten as a $\FOPZ(\exists\forall\exists)$-formula, see the full version of the paper for details.

The Hausdorff distance under $m$ Translations occurs naturally when approximating the Hausdorff distance under translation: Specifically, common algorithms compute a set~$A$ of $|A|= f(\epsilon)$ translations such that $\delta_{\overrightarrow{H}}^{T(A)}(B,C) \le (1+\epsilon)\delta_{\overrightarrow{H}}^{T}(B,C)$.
Generally, this problem is then solved by performing $|A|$ computations of the Hausdorff distance, which yields $\tOh(|A|n)=\tOh(f(\epsilon)n)$-time algorithms \cite{Wenk03}. Improving over the $\tOh(m n)$-time baseline for Hausdorff Distance under $m$ Translations would thus lead to immediate improvements for approximating the Hausdorff Distance under Translation.
Our results will establish additional consequences of fast algorithms for this problem: an $O(n^{2-\epsilon_d})$-time algorithm with $\epsilon_d > 0$ for Hausdorff distance under $n$ Translations would give an algorithmic improvement for the classes of $\FOPZ(\exists \forall \exists)$- and $\FOPZ(\forall \exists \forall)$-formulas.

\paragraph*{Verification of Pareto Sums}

Our second geometric problem is a verification version of computing \emph{Pareto sums}: Given point sets $A,B \subseteq \mathbb{Z}^{d}$, the Pareto sum $C$ of $A,B$ is defined as the Pareto front of their sumset $A+B=\{a+b\mid a\in A,b\in B\}$. Put differently, the Pareto sum of $A,B$ is a set of points $C$ satisfying (1) $C \subseteq A+B$,  (2) for every $a \in A$ and $b \in B$, the vector $a+b$ is dominated\footnote{We consider the usual domination notion: A vector $u\in \mathbb{Z}^d$ is dominated by some vector $v\in \mathbb{Z}^d$ (written $u \leq v$) if and only if in all dimensions $i\in [d]$ it holds that $u[i]\leq v[i]$.} by some $c \in C$ and (3) there are no distinct $c,c' \in C$ such that $c'$ dominates $c$. 
The task of computing Pareto sums appears in various multicriteria optimization settings~\cite{artigues2013state,schulze2019multi,ehrgott2000survey,lust2014variable}; fast output-sensitive algorithms (both in theory and in practice) have recently been investigated by Hespe, Sanders, Storandt, and Truschel~\cite{DBLP:conf/esa/Hespe0ST23}. 

We consider the following problem as \emph{Pareto Sum Verification}: Given $A,B,C\subseteq \mathbb{Z}^d$, determine whether
\begin{align*}
\forall a \in A \forall b \in B \exists c \in C: a+b \le c.
\end{align*}
The complexity of Pareto Sum Verification\footnote{We remark that our problem definition only checks a single of the three given conditions, specifically, condition~(2). However, in Section~\ref{sec:ParetoSum}, we will establish that the verifying \emph{all three} conditions reduces to verifying this single condition. More specifically, for sets $A,B,C$ of size at most $n$, we obtain that if we can solve Pareto Sum Verification in time $T(n)$, then we can check whether $C$ is the Pareto sum of $A,B$ in time $O(T(n))$.} is tightly connected to output-sensitive algorithms for Pareto Sum. Specifically, solving Pareto Sum Verification reduces to \emph{computing} the Pareto sum $C$ when given inputs $A,B$ of size at most $n$ with the promise that $|C| =\Theta(n)$; see Section~\ref{sec:ParetoSum} for details. The work of Hespe et al.~\cite{DBLP:conf/esa/Hespe0ST23} gives a practically fast $O(n^2)$-time algorithm in this case for $d=2$; note that for $d\ge 3$, we still obtain an $\tOh(n^2)$-time algorithm via our Baseline Algorithm, which is described in the full version of the paper.

\paragraph{A problem pair that is complete for $\FOPZ$}
As a pair, these two geometric problems turn out to be fine-grained complete for the class $\mathsf{FOP}_{\mathbb{Z}}$.
  \begin{restatable}{theorem}{completenessclass}
	There is a function $\epsilon(d)>0$ such that 
	both of the following problems can be solved in time $O(n^{2-\epsilon(d)})$
	\begin{itemize}
	\item Pareto Sum Verification,
	\item Hausdorff distance under $n$ Translations,
	\end{itemize}
	if and only if for each problem $P$ in $\mathsf{FOP}_{\mathbb{Z}}^k$ with $k\geq3$ there exists an $\epsilon_P>0$ such that $P$ can be solved in time $O(n^{k-1-\epsilon_P})$. 
	\label{completenesswholeFOP3}
	\end{restatable}

The above theorem shows that a single pair of natural problems captures the fine-grained complexity of the expressive and diverse class $\FOPZ$. As an illustration just how expressive this class is, we observe the following barriers:\footnote{The first three statements follow from $\FOPZ$ generalizing the class $PTO$ studied in~\cite{DBLP:journals/algorithmica/AnGIJKN22}, see Section~\ref{sec:relatedWork}. The remaining statements rely on the additive structure of $\FOPZ$.}
\begin{enumerate}
	\item If there is some $\epsilon >0$ such that all problems in $\FOPZ(\exists \exists \forall)$ (or $\FOPZ(\forall \forall \exists)$) can be solved in time $O(n^{2-\epsilon})$, then OVH (and thus SETH) is false \cite[Theorem 16]{DBLP:journals/algorithmica/AnGIJKN22}.
	\item If there is some $\epsilon >0$ such that all problems in $\FOPZ(\exists \forall \exists)$ (or $\FOPZ(\forall \exists \forall)$) can be solved in time $O(n^{2-\epsilon})$, then the Hitting Set Hypothesis is false \cite[Theorem 12]{DBLP:journals/algorithmica/AnGIJKN22}.
	\item If for all problems $P$ in $\FOPZ(\exists \exists \forall)$ (or $\FOPZ(\forall \forall \exists)$), there exists some $\epsilon>0$ such that we can solve $P$ in  $O(n^{2-\epsilon})$, then the 3-uniform Hyperclique Hypothesis is false \cite[Theorem 15]{DBLP:journals/algorithmica/AnGIJKN22}.
	\item If for all problems $P$ in $\FOPZ(\exists \exists \exists)$ ($\FOPZ(\forall \forall \forall), \FOPZ(\forall \forall \exists)$, or $\FOPZ(\exists \exists \forall)$), there exists some $\epsilon>0$ such that we can solve $P$ in time $O(n^{2-\epsilon})$, then the $3$-SUM Hypothesis is false (Theorem \ref{existential_complete} with Lemma \ref{verif-complete-three}).
\end{enumerate}
Theorem \ref{completenesswholeFOP3} raises the question whether for any constant dimension $d$, the Hausdorff distance under $n$ Translations admits a subquadratic reduction to Pareto Sum Verification. A positive answer would establish Pareto Sum Verification as complete for the \emph{entire} class $\FOPZ$. We elaborate on this in Section~\ref{sec::ende}.

\subsubsection{$3$-SUM is complete for $\mathsf{FOP}_{\mathbb{Z}}$ formulas of low inequality dimension}

Returning to our motivating question, we ask: Since it appears unlikely to prove completeness of 3-SUM for all $\FOPZ$ formulas (as this requires a tight 3-uniform hyperclique lower bound for 3-SUM), can we at least identify a large fragment of $\FOPZ$ for which $3$-SUM is complete? In particular, can we extend our first result of Theorem~\ref{existential_complete} from existentially quantified formulas to substantially different problems in $\FOPZ$, displaying other quantifier structures?
 
Surprisingly, we are able to show that $3$-SUM is complete for \emph{low-dimensional} $\FOPZ$ formulas, \emph{independent of their quantifier structure}.
To formalize this, we introduce the \emph{inequality dimension} of a $\FOPZ$ formula as the smallest number of linear inequalities required to model it. More formally, consider a $\FOPZ$ formula $\phi = Q_1 x_1\in A_1, \dots, Q_k x_k\in A_k: \psi$ with $A_i\subseteq \mathbb{Z}^{d_i}$. The \emph{inequality dimension} of $\phi$ is the smallest number $s$ such that there exists a Boolean function $\psi' :\{0,1\}^s \to \{0,1\}$ and (strict or non-strict) linear inequalities $L_1, \dots, L_s$ in the variables $\{x_i[j] : i\in \{1,\dots,k\} ,j\in \{1,\dots,d_i\} \}$ and the free variables such that $\psi(x_1,\dots, x_k)$ is equivalent to $\psi'(L_1,\dots, L_s)$. As an example, the 3-SUM formula $\exists a\in A \exists b\in B\exists c\in C: a+b=c$ has inequality dimension 2, as $a+b=c$ can be modelled as conjunction of the two linear inequalities $a+b \le c$ and $a+b \ge c$, whereas no single linear inequality can model $a+b=c$.

We show that $3$-SUM is fine-grained complete for model-checking $\mathsf{FOP}^3_{\mathbb{Z}}$ formulas with inequality dimension at most $3$.
This result is our perhaps most interesting technical contribution and intuitively combines our result that $3$-SUM is hard for counting $\FOPZ$ witnesses (Corollary~\ref{Count3COMP}) with a geometric argument, specifically, that the union of $n$ unit cubes in $\mathbb{R}^3$ can be decomposed into the union of $O(n)$ pairwise interior- and exterior-disjoint axis-parallel boxes. To this end, we extend a result from \cite{DBLP:journals/dcg/ChewDEK99}, which constructs pairwise interior-disjoint axis-parallel boxes, to also achieve exterior-disjointness. For more details, see the Technical Overview below and Section~\ref{sec:IneqDimension}. 
\begin{restatable}[]{theorem}{threesumineq}
	There is an algorithm deciding $3$-SUM in randomized time $O(n^{2-\epsilon})$ for an $\epsilon>0$,
	if and only if for each problem $P$ in $\mathsf{FOP}^{k}_{\mathbb{Z}}$ with $k\geq 3$ and inequality dimension at most $3$, there exists some $\epsilon > 0$ such that we can solve $P$ in randomized time $O(n^{k-1-\epsilon})$.
	\label{three-sum-Completeness-all-quantifer}
\end{restatable}

Note that this fragment of $\FOPZ$ contains a variety of interesting problems. A general example is given by comparisons of sets defined using the sumset arithmetic\footnote{The sumset arithmetic uses the sumset operator $X+Y$ to denote the sumset $\{x+y\mid x\in X, y\in Y\}$ and $\lambda X$ to denote $\{\lambda x \mid x\in X\}$.}, which correspond to formulas of inequality dimension at most 2: E.g., checking, given sets $A,B,C\subseteq \mathbb{Z}$ and $t\in \mathbb{Z}$, whether $C$ is an additive $t$-approximation of the sumset $A+B$ is equivalent to verifying the conjunction of the $\FOPZ(\forall \forall \exists)$ problem\footnote{Note that the corresponding formula is $\forall a\in A\forall b\in B \exists c\in C: (c\le a+b) \wedge (a+b \le c+t)$, which clearly has inequality dimension at most 2.} $A+B \subseteq C+\{0,\dots,t\}$ and (2) the $\FOPZ(\forall \exists \exists)$ problem\footnote{Note that the corresponding formula is $\forall c\in C\exists a\in A\exists b\in B: a+b = c$, which clearly has inequality dimension at most 2.} $C\subseteq A+B$. Likewise, this extends to $\lambda$-multiplicative approximations of sumsets. Furthermore, the problems corresponding to general sumset comparisons like $\alpha_1 A_1 + \cdots + \alpha_i A_i \subseteq \alpha_{i+1} A_{i+1} + \cdots + \alpha_k A_{k} + \{-\ell, \dots, u\}$ have inequality dimension at most~$2$ as well.

Our results of Theorems~\ref{completenesswholeFOP3} and~\ref{three-sum-Completeness-all-quantifer} suggests to view Pareto Sum Verification as a geometric, high-dimensional generalization of $3$-SUM. Furthermore, it remains an interesting problem to establish the highest $d$ such that 3-SUM is complete for $\FOPZ$ formulas of inequality dimension at most $d$; for a discussion see Section \ref{sec::ende}.

\paragraph*{Further Applications}
As an immediate application of our first completeness theorem, we obtain a simple proof of a $n^{4/3-o(1)}$ lower bound for the 4-SUM problem based on the the $3$-uniform hyperclique hypothesis; see the full version of the paper for details. Specifically, by Theorem~\ref{existential_complete}, it suffices to model the 3-uniform 4-hyperclique problem as a problem in $\mathsf{FOP}_{\mathbb{Z}}(\exists \exists \exists \exists)$. The resulting conditional lower bound is implicitly known in the literature, as it can alternatively be obtained by combining a 3-uniform hyperclique lower bound for $4$-cycle given in~\cite{DBLP:conf/soda/LincolnWW18} with a folklore reduction from $4$-cycle to $4$-SUM (see~\cite{DBLP:journals/algorithmica/JafargholiV16} for a deterministic reduction from $3$-cycle to $3$-SUM).
\begin{restatable}{theorem}{lowerbound} \label{thm:4SUMlb}
    If there is some $\epsilon > 0$ such that 4-SUM can be solved in time $O(n^{\frac{4}{3}-\epsilon})$, then the 3-uniform hyperclique hypothesis fails. 
\end{restatable}
 Similarly, we can also give a simple proof for a known lower bound for $3$-SUM.

Another application of our results is to establish class-based conditional bounds. As a case in point, consider the problem of computing the Pareto sum of $A,B\subseteq \mathbb{Z}^d$: Clearly, this problem can be solved in time $\tOh(n^2)$ by explicitly computing the sumset $A+B$ and computing the Pareto front using any algorithm running in near-linear time in its input, e.g.~\cite{DBLP:conf/stoc/GabowBT84}. We prove the following conditional optimality results already in the case when the desired output (the Pareto sum of $A,B$) has size $\Theta(n)$. 
\begin{restatable}[Pareto Sum Computation Lower Bound]{theorem}{lowerboundPS}
	The following conditional lower bounds hold for output-sensitive Pareto sum computation:
	\begin{enumerate}
		\item If there is $\epsilon > 0$ such that we can compute the Pareto sum $C$ of $A,B\subseteq \mathbb{Z}^2$, whenever $C$ is of size $\Theta(n)$, in time $O(n^{2-\epsilon})$, then the $3$-SUM hypothesis fails (thus, for any $\FOPZ^k$ formula $\phi$ of inequality dimension at most 3, there is $\epsilon'>0$ such that $\phi$ can be decided in time $O(n^{k-1-\epsilon'})$). 
		\item If for all $d\ge 2$, there is $\epsilon > 0$ such that we can compute the Pareto sum $C$ of $A,B\subseteq \mathbb{Z}^d$, whenever $C$ is of size $\Theta(n)$, in time $O(n^{2-\epsilon})$, then there is some $\epsilon' > 0$ such that we can decide all $\FOPZ$ formulas with $k$ quantifiers not ending in $\exists \forall \exists$ or $\forall \exists \forall$ in time $O(n^{k-1-\epsilon'})$. 
	\end{enumerate}
	\end{restatable}

Our lower bound for 2D strengthens a quadratic-time lower bound found by Funke et al.~\cite{HespeSST24-PC} based on the $(\min,+)$-convolution hypothesis to hold already under the weaker (i.e., more believable) $3$-SUM hypothesis. For higher dimensions, we furthermore strengthen the conditional lower bound via its connection to $\FOPZ$.

We conclude with remaining open questions in Section \ref{sec::ende}.

\subsection{Further Related Work}\label{sec:relatedWork}
To our knowledge, the first investigation of the connection between classes of model-checking problems and central problems in fine-grained complexity was given by Williams~\cite{DBLP:conf/csl/Williams14}, who shows that the $k$-clique problem is complete for the class of existentially-quantified first order graph properties, among other results. As important follow-up work, Gao et al.~\cite{GaoIKW19} establish OV as complete problem for model-checking any first-order property.

Subsequent results include classification results for $\exists^k \forall$-quantified first-order graph properties~\cite{DBLP:conf/coco/BringmannFK19}, fine-grained upper and lower bounds for counting witnesses of first-order properties~\cite{DBLP:conf/icalp/DellRW19}, completeness theorems for multidimensional ordering properties~\cite{DBLP:journals/algorithmica/AnGIJKN22} (discussed below), completeness and classification results for optimization classes~\cite{DBLP:conf/approx/BringmannCFK21,DBLP:conf/icalp/BringmannCFK22} as well as an investigation of sparsity for monochromatic graph properties~\cite{FischerKR24}.

We remark that An et al.~\cite{DBLP:journals/algorithmica/AnGIJKN22} study completeness results for a strict subset of $\FOPZ$ formulas: Specifically, they introduce a class $\PTO_{k,d}$ of $k$-quantifier first-order sentences over inputs $\mathbb{N}^d$ (or, without loss of generality $\{1,\dots, n\}^d$) that may only use \emph{comparisons} of coordinates (and constants).
Note that such sentences lack additive structure, and indeed the fine-grained complexity differs decisively from $\FOPZ$: E.g., for $\PTO(\exists \exists \exists)$ formulas, they establish the sparse triangle detection problem as complete,  establishing a conditionally tight running time of $m^{2\omega/(\omega+1)\pm o(1)}$.
This is in stark contrast to $\FOPZ(\exists \exists \exists)$ formulas, for which we establish $3$-SUM as complete problem, yielding a conditionally optimal running time of $n^{2\pm o(1)}$. In particular, for each 3-quantifier structure $Q_1Q_2Q_3$, a $O(n^{2-\epsilon})$-time algorithm for all $\FOPZ(Q_1Q_2Q_3)$ problems would break a corresponding hardness barrier\footnote{Specifically, an $O(n^{2-\epsilon})$ time algorithm for problems in $\mathsf{FOP}_{\mathbb{Z}}(\exists \exists \exists), \FOPZ(\forall \forall \forall), \mathsf{FOP}_{\mathbb{Z}}(\forall \forall \exists)$, or $\FOPZ(\exists \exists \forall)$ with $\epsilon>0$ would
refute the $3$-SUM hypothesis.
Furthermore, an $O(n^{2-\epsilon})$ time algorithm for problems in $\mathsf{FOP}_{\mathbb{Z}}(\forall \exists \exists)$, $\FOPZ(\exists \forall \forall)$, $\mathsf{FOP}_{\mathbb{Z}}(\exists \forall \exists)$, or $\FOPZ(\forall \exists \forall)$ with $\epsilon>0$ would
immediately yield an improvement for the MaxConv lower bound problem \cite{DBLP:journals/talg/CyganMWW19}; for details see the full version of the paper.}.

Since any $\PTO_{k,d}$ formula is also a $\FOPZ$ formula with the same quantifier structure, any hardness result in~\cite{DBLP:journals/algorithmica/AnGIJKN22} for $\PTO(Q_1, \dots, Q_k)$ carries over to $\FOPZ(Q_1,\dots, Q_k)$. On the other hand, any of our algorithmic results for $\FOPZ(Q_1,\dots, Q_k)$ transfers to its subclass $\PTO(Q_1,\dots, Q_k)$.

\section{Technical Overview}
\label{sec:TechnicalOverview}

In this section, we sketch the main ideas behind our proofs.
\paragraph*{Completeness of $k$-SUM for $\mathsf{FOP_{\mathbb{Z}}}(\exists^k)$}
With the right ingredients, proving that $k$-SUM is complete for $\FOPZ$ formulas with $k$ existential quantifiers (Theorem~\ref{existential_complete}) is possible via a simple approach: We observe that any $\FOPZ(\exists^k)$ formula $\phi$ can be rewritten such that we may assume that $\phi$  is a conjunction of $m$ inequalities. We then use a slight generalization of a bit-level trick of~\cite{DBLP:journals/siamcomp/WilliamsW13} to reduce each inequality to an equality, incurring only $O(\log n)$ overhead per inequality (intuitively, we need to guess the most significant bit position at which the left-hand side and the right-hand side differ).
Thus, we obtain $O(\log^m n)$ conjunctions of $m$ equalities; each such conjunction can be regarded as an instance of Vector $k$-SUM. Using a straightforward approach for reducing Vector $k$-SUM to $k$-SUM given in~\cite{DBLP:journals/corr/AbboudLW13}, the reduction to $k$-SUM follows. We give all details in Section~\ref{sec:existentialpa} and the full version of the paper.

\paragraph*{Counting witnesses and handling multisets}
While the reduction underlying Theorem \ref{existential_complete} preserves the existence of solutions, it fails 
to preserve the number of solutions. The challenge is that when applying the bit-level trick to reduce inequalities to equalities, we need to make sure that for each witness of a $\FOPZ(\exists^k)$ formula $\phi$, there is a unique witness in the $k$-SUM instances produced by the reduction. While it is straightforward to ensure that we do not produce multiple witnesses, the subtle issue arises that distinct witnesses for $\phi$ may be mapped to the same witness in the $k$-SUM instances. It turns out that it suffices to solve a \emph{multiset} version of \#$k$-SUM, i.e., to count all witnesses in a $k$-SUM instance in which each input number may occur multiple times. 

Thus, to obtain Theorem~\ref{thm:counting-witnesses}, we show a fine-grained equivalence of Multiset \#$k$-SUM and \#$k$-SUM, for all odd $k\ge 3$. This fine-grained equivalence, which we prove via a heavy-light approach, might be of independent interest.\footnote{We remark that it is plausible that the proof of the subquadratic equivalence of $3$-SUM and \#$3$-SUM due to Chan et al.~\cite{DBLP:journals/corr/abs-2303-14572} could be extended to establish subquadratic equivalence with Multiset \#$3$-SUM as well. Note, however, that a fine-grained equivalence of \#$k$-SUM and $k$-SUM is not known for any $k\ge 4$.}
Combining this equivalence with an inclusion-exclusion argument, we may thus lift Theorem~\ref{existential_complete} to a counting version for all odd $k \ge 3$.

In the reductions below, we will make crucial use of the immediate corollary of Theorem~\ref{thm:counting-witnesses} and~\cite{DBLP:journals/corr/abs-2303-14572} that for each $\FOPZ(\exists \exists \exists)$ formula $\phi$, there exists a subquadratic reduction from counting witnesses for $\phi$ to $3$-SUM (Corollary~\ref{Count3COMP}).

\paragraph*{On general quantifier structures}
We perform a systematic study on the different quantifier structures for $k=3$. 
Due to simple negation arguments, we only have to perform a systematic study on the classes of problems
$\mathsf{FOP}_{\mathbb{Z}}(\exists \exists \exists)$,
$\mathsf{FOP}_{\mathbb{Z}}(\forall \exists \exists)$,
$\mathsf{FOP}_{\mathbb{Z}}(\forall \forall \exists)$,
$\mathsf{FOP}_{\mathbb{Z}}(\exists \forall \exists)$.

First, we state a simple lemma establishing syntactic complete problems for the classes above.
\begin{lemma}[Syntactic Complete problems (Informal Version)]
Let $Q_1,Q_2 \in \{\exists,\forall\}$. We can reduce every formula of the class $\mathsf{FOP}_{\mathbb{Z}}(Q_1Q_2\exists)$ to the formula  
$$Q_1 \Tilde{a_1} \in \Tilde{A_1} Q_2 \Tilde{a_2}  \in \Tilde{A_2} \exists \Tilde{a_3}  \in \Tilde{A_3}: \Tilde{a_1} +\Tilde{a_2}  \leq \Tilde{a_3}.  $$
\end{lemma}

\paragraph*{On the quantifier change $\mathsf{FOP}_{\mathbb{Z}}(\forall \exists \exists) \to \mathsf{FOP}_{\mathbb{Z}}(\exists \exists \exists) $.}
We rely on the subquadratic equivalence between $3$-SUM and a functional version of $3$-SUM called All-ints $3$-SUM, which asks to determine 
for every $a \in A$  whether there is a solution involving $a$. A randomized subquadratic equivalence was given in~\cite{DBLP:conf/focs/WilliamsW10}, which can be turned deterministic~\cite{DBLP:conf/icalp/LincolnWWW16}.

This equivalence allows us to use the bit-level trick to turn inequalities to equalities, despite it seemingly not interacting well with the quantifier structure $\forall \exists \exists$ at first sight.
This results in a proof of the following hardness result.
\begin{restatable}{lemma}{allintshard}
	If $3$-SUM can be solved in time $O(n^{2-\epsilon})$ for an $\epsilon>0$,
	then all problems $P$ of $\mathsf{FOP}_{\mathbb{Z}}(\forall \exists \exists)$
	can be solved in time $O(n^{2-\epsilon_{P}})$ for an $\epsilon_{P}>0$. 
	\label{fopaee}
	\end{restatable}

\paragraph*{On the quantifier change $\mathsf{FOP}_{\mathbb{Z}}(\exists \exists \exists) \to \mathsf{FOP}_{\mathbb{Z}}(\forall \forall \exists) $.}
As a first result for the class $\mathsf{FOP}_{\mathbb{Z}}(\forall \forall \exists)$, 
we are able to show equivalence to $3$-SUM for a specific problem in this class,
thus introducing a $3$-SUM equivalent problem
with a different quantifier structure in comparison to $3$-SUM.
Specifically, we consider the problem of verifying additive $t$-approximation of sumsets.
We are able to precisely characterize the fine-grained complexity depending on $t$.

Formally, we show the following theorem.
\begin{restatable}{theorem}{sumsetapproxchar}
	Consider the Additive Sumset Approximation problem of deciding, given $A,B,C\subseteq \mathbb{Z}, t\in \mathbb{Z}$, whether
	$$A+B \subseteq C +\{0,\dots,t\}.$$
	This problem is
	\begin{itemize}
	\item solvable in time $O(n^{2-\delta})$ with $\delta>0$, whenever $t=O(n^{1-\epsilon})$ for any $\epsilon>0,$
	\item not solvable in time $O(n^{2-\epsilon})$, whenever $t =\Omega(n)$ assuming the Strong $3$-SUM hypothesis.
	\end{itemize}
	  Furthermore, subquadratic hardness holds under the standard 3-SUM Hypothesis if no restriction on $t$ is made.
	  \label{sumsetapproxTHM}
	\end{restatable}

The above theorem is essentially enabling a quantifier change transforming the $\exists \exists \exists$ quantifier structure for which $3$-SUM is complete 
into a subquadratic equivalent problem with a quantifier structure $\forall \forall \exists$.
Moreover, the $3$-SUM hardness is a witness to the hardness of the class $\mathsf{FOP}_{\mathbb{Z}}(\forall \forall \exists)$.

Let us remark a few interesting aspects: The algorithmic part follows from sparse convolution techniques
going back to Cole and Hariharan~\cite{DBLP:conf/stoc/ColeH02}, see~\cite{DBLP:journals/corr/abs-2107-07625} for a recent account and also \cite{DBLP:conf/stoc/ChanL15,DBLP:conf/icalp/BringmannN21,DBLP:conf/stoc/BringmannFN21}. Specifically, whenever $t = O(n^{1-\epsilon})$, it holds that $|C + \{0, \dots, t\}| = O(n^{2-\epsilon})$ and intuitively,
we can use an output-sensitive convolution algorithm to compute $A+B$ and compare it to $C+\{0, \dots, t\}$.\footnote{The argument is slightly more subtle, since we need to avoid computing $A+B$ if its size exceeds $O(n^{2-\epsilon})$.}
Our result indicates that an explicit construction of $C+\{0, \dots, t\}$ is required, since once it may get as large as $\Omega(n^2)$, we obtain a $n^{2-o(1)}$-time lower
bound assuming the Strong $3$-SUM Hypothesis.  

The lower bound follows from describing the $3$-SUM problem alternatively as $(A+B) \cap C \neq \emptyset$, which is equivalent to the negation of $(A+B)\subseteq \Bar{C}$, where $\Bar{C}$ denotes the complement of $C$.
Thus, we aim to cover the complement of $C$ by intervals of length $t$. While this appears impossible for $3$-SUM, we employ the subquadratic equivalence of $3$-SUM and its convolutional version due to Patrascu \cite{DBLP:conf/stoc/Patrascu10}. This problem will deliver us the necessary structure to represent this complement with the addition of few auxilliary points.

The reverse reduction from Additive Sumset Approximation to $3$-SUM follows from Theorem~\ref{three-sum-Completeness-all-quantifer} (as Additive Sumset Approximation has inequality dimension $2$).

\paragraph*{On completeness results for $\mathsf{FOP}_{\mathbb{Z}}^k$}
The above ingredients establish our completeness theorems by exhaustive search over remaining quantifiers. Specifically, by a combination of Theorem~\ref{sumsetapproxTHM}, which shows that Additive Sumset Approximation is $3$-SUM hard, and a combination of Lemma~\ref{fopaee} and Theorem~\ref{existential_complete},
we get:
\begin{restatable}[]{lemma}{verifhard}
There is a function $\epsilon(d)>0$ such that 
the Verification of Pareto Sum problem can be solved in time $O(n^{2-\epsilon(d)})$
if and only if all problems $P$ in the classes 
\begin{itemize}
\item $\mathsf{FOP}_{\mathbb{Z}}(Q_1\dots Q_{k-3}\exists \exists \exists)$,$\mathsf{FOP}_{\mathbb{Z}}(Q_1\dots Q_{k-3}\forall \forall \forall),$
\item $\mathsf{FOP}_{\mathbb{Z}}(Q_1\dots Q_{k-3}\forall \exists \exists)$,$\mathsf{FOP}_{\mathbb{Z}}(Q_1\dots Q_{k-3}\exists \forall \forall),$
\item $\mathsf{FOP}_{\mathbb{Z}}(Q_1\dots Q_{k-3}\forall \forall \exists)$,$\mathsf{FOP}_{\mathbb{Z}}(Q_1\dots Q_{k-3}\exists \exists \forall),$
\end{itemize}
where $Q_1, \dots Q_{k-3} \in \{ \exists, \forall \}$ and $k\geq 3$,
can be solved in time $O(n^{k-1-\epsilon_P})$ for an $\epsilon_P>0$.
\label{verif-complete-three}
\end{restatable}

Similarly, for quantifier structures ending in $\exists \forall \exists$ and $\forall \exists \forall$, we obtain the following completeness result.

\begin{restatable}[]{lemma}{hunthard}
	There is a function $\epsilon(d)>0$ such that 
	the Hausdorff Distance under $n$ Translations problem can be solved in time $O(n^{2-\epsilon(d)})$
	if and only if all problems $P$ in the classes
	\begin{itemize}
    \item $\mathsf{FOP}_{\mathbb{Z}}(Q_1\dots Q_{k-3}\exists \forall \exists), \mathsf{FOP}_{\mathbb{Z}}(Q_1\dots Q_{k-3}\forall \exists \forall),$
	\end{itemize}
	where $Q_1, \dots Q_{k-3} \in \{ \exists, \forall \}$ and $k\geq 3$,
	can be solved in time $O(n^{k-1-\epsilon_P})$ for an $\epsilon_P>0$.
	\label{hunthard}
	\end{restatable}

The combination of Lemma \ref{verif-complete-three} and Lemma \ref{hunthard}, thus suffice to 
prove Theorem \ref{completenesswholeFOP3}.
\paragraph*{The $3$-SUM completeness of formulas with inequality dimension at most $3$}

As a first idea, one could try to solve problems of different quantifier structures
by just counting witnesses. Consider in the following the example $\FOPZ(\forall \forall \exists)$. 

Assume we are promised that the formula $\forall a\in A \forall b\in B \exists c\in C \psi(a,b,c)$ 
satisfies a kind of \emph{disjointness} property, specifically that for every $(a,b) \in A \times B$ there exists at most one $c \in C$ such that $\psi(a,b,c)$. Then satisfying the formula boils down to checking whether the number of witnesses $(a,b,c)$ satisfiying $\psi(a,b,c)$ equals to $|A|\cdot |B|$.

To create this \emph{disjointness} effect, we use the following geometric approach: 
We show that one can re-interpret the formula as $\forall a\in A \forall b\in B: a+b\in \bigcup_{c'\in C'} V(c')$, where $A,B,C' \subseteq \mathbb{Z}^3$, $C'$ is a set of size $O(n)$ and $V(c')$ is an orthant associated to $c'$. Using an adapted variant of~\cite{DBLP:journals/dcg/ChewDEK99}, we decompose this union of orthants in $\mathbb{R}^3$ (which we may equivalently view as sufficiently large congruent cubes) into a set $\mathcal{R}$ of $O(n)$ \emph{disjoint} boxes. Thus, it remains to notice that the resulting problem  -- i.e., for all $a\in A, b\in B$ is there a box $R\in \mathcal{R}$ such that $a+b$ is contained in $R$ -- is a $\FOPZ(\forall \forall \exists)$ formula with the desired disjointness property, which can be handled as argued above.  
For the class $\mathsf{FOP}_{\mathbb{Z}}(\exists \forall \exists)$, we perform a slightly more involved argument.
The classes $\mathsf{FOP}_{\mathbb{Z}}(\exists \exists \exists)$ and $ \mathsf{FOP}_{\mathbb{Z}}(\forall \exists \exists)$
reduce to $3$-SUM regardless of the inequality dimension due to Theorem \ref{existential_complete} and Lemma \ref{fopaee}.

\section{$k$-SUM is complete for existential $\FOPZ$ formulas}\label{sec:existentialpa}
We begin with a simple completeness theorem that $k$-SUM is complete for the class of problems $\mathsf{FOP}_{\mathbb{Z}}(\exists^k)$. Since $k$-SUM is indeed a $\FOPZ(\exists^k)$-formula, it remains to show a fine-grained reduction from any $\FOPZ(\exists^k)$ formula to $k$-SUM.
The proofs in this section are deferred to the full version.
As a first step towards this Theorem, we consider how to reduce a conjunction of $m$ linear inequalities
to a vector $k$-SUM instance.
\begin{lemma}
  Consider vectors $a_1 \in \{-U, \dots, U\}^{d_1}, \dots, a_k \in \{-U ,\dots,U\}^{d_k},$ integers $S_1, \dots, S_m \in \{-U, \dots, U\},$
  for each $i \in \{1,\dots,m\}, j \in \{1,\dots,k\}$, vectors $c_{i,j} \in \mathbb{Z}^{d_j}$, and a formula $$\psi:=\bigwedge_{i=1}^{m}\left( \sum_{j=1}^{k}c_{i,j}^T a_j \geq S_i\right). $$
  There exist  $O(1)$ time computable functions $f_1^{\ell,\psi} ,\dots ,f_k^{\ell,\psi} ,g^{\ell,\psi,W}$ such that the following statements are equivalent
  \begin{enumerate}
    \item The formula $\bigwedge_{i=1}^{m} \left(\sum_{j=1}^{k}c_{i,j}^T a_j \geq S_i \right)$ holds.
    \item There are $\ell \in \{1,\dots, \lceil \log_2(M)\rceil\}^m, W \in \{1,\dots ,k\}^m$ such that $f^{\ell,\psi}_1 (a_1)+\dots + f^{\ell,\psi}_k (a_k)=g^{\ell,\psi,W} (S_1,\dots,S_m).$
  \end{enumerate}
  Moreover, if the second item holds, there is a unique choice of such $\ell$ and $W$. 
      \label{Ineq_to_eq}
\end{lemma}
Essentially the above lemma enables a reduction from a conjunction of inequality checks to a conjunction of equality checks. We can now continue with our completeness theorem.
\threesumcompl*
Ater applying Lemma~\ref{Ineq_to_eq}, it remains to reduce a conjunction of equality checks to $k$-SUM. To do so, we interpret the conjunction of equalities as a Vector $k$-SUM problem, which can be reduced to $k$-SUM in a straightforward way~\cite{DBLP:journals/corr/AbboudLW13}.

\section{On counting witnesses in $\mathsf{FOP}_{\mathbb{Z}}$}\label{sec:countwitn}
  In this section, we show reductions from counting witnesses of $\FOPZ(\exists^k)$ formulas to $\#k$-SUM, specifically, we prove Theorem~\ref{thm:counting-witnesses}. To do so, we adapt the proof of Theorem~\ref{existential_complete} given in Section~\ref{sec:existentialpa} to a counting version. As discussed in Section~\ref{sec:TechnicalOverview}, this requires us to work with a multiset version of \#$k$-SUM. Handling multisets is thus the main challenge addressed in this section.
   Formally, we say that a multiset is a set $A$ together with a function $f:A\to \mathbb{N}$. For $a \in A$, we abbreviate $n_a:= f(a)$ as the multiplicity
   of $a$. To measure multiset sizes, we still think of each $a$ to have $n_a$ copies in the input, i.e. the size of $A$ is $\sum_{a \in A}n_a$.
   Almost all proofs in this section are deferred to the full version of the paper.
   \begin{definition}[$(U,d)$-vector Multiset $\#k$-SUM]
   Let $ X:=\{-U,\dots,U\}^d$. Given $k$ multisets $A_1,\dots,A_k \subseteq X$ and $t \in X$, we ask for the total number of $k$-SUM witnesses, that is 
   \begin{align*}
   \sum_{\substack{a_1+\dots +a_k=t, \\ a_1 \in A_1, \dots, a_k \in A_k}}\prod_{i=1}^k n_{a_i}.
   \end{align*}
   \end{definition}
   Furthermore, define Multiset $\#k$-SUM as $(U,1)$-vector Multiset $\# k$-SUM and $M$-multiplicity $\# k$-SUM as Multiset $\#k$-SUM with the additional restriction that the multiplicity of each element is limited, that is for all $a \in A_1 \cup \dots \cup A_k: n_a \leq M$ holds. 
   Lastly, $\# k$-SUM is defined as $1$-Multiplicity $\# k$-SUM and $(U,d)$-vector $\#k$-SUM is $(U,d)$-vector Multiset $\#k$-SUM where for all $a \in A_1 \cup \dots \cup A_k: n_a = 1$ holds.

   For the case of $\FOPZ^3$ we will also introduce the $\#$All-ints version of the above problems, which asks to determine, for each $a_1 \in A_1$, the number of witnesses involving $a_1$.

The (deferred) proof of the following lemma is analogous to the proof of Abboud et al.~\cite{DBLP:journals/corr/AbboudLW13} to reduce Vector $k$-SUM to $k$-SUM.
\begin{lemma}[$(U,d)$-vector Multiset $\#k$-SUM $\leq_{\lceil k/2 \rceil }$ Multiset $\#k$-SUM]
If Multiset $\#k$-SUM can be solved in time~$T(n)$ then $(U,d)$-vector Multiset $\#k$-SUM can be solved in time $O(nd \log(U)+T(n)).$
\label{vector_to_one_red}
\end{lemma}

Next, we give a simple approach to solve Multiset $\# k$-SUM when all multiplicities are comparably small.

\begin{lemma}[$M$-multiplicity $\# k$-SUM$\leq_{\lceil k/2 \rceil}$ $\# k$-SUM]
  If $\#k$-SUM can be solved in time~$T(n)$, then $M$-multiplicity $\#k$-SUM can be solved in time $\Tilde{O}(T(nM^{k-1}))$.
\label{reduction_t}
\end{lemma}

For later purposes, we will need the following version of the above lemma.
\begin{observation}
If $\#$All-ints $3$-SUM can be solved in time $T(n)$,
then we can solve $\#$All-ints $M$-multiplicity $3$-SUM in time $\Tilde{O}(T(nM^{2})).$
\label{All-ints-count}
\end{observation}

We can finally prove the main result of this section. 

\begin{lemma}
    For odd $k\geq 3$, if there exists an algorithm for the $\#k$-SUM problem running in time $O(n^{\lceil k/2 \rceil-\epsilon})$ for an $\epsilon>0$, then there exists an algorithm for the Multiset $\#$k-SUM problem running in time $O(n^{\lceil k/2 \rceil-\epsilon'})$ for an $\epsilon'>0$.
  \label{multisetToSet}  
  \end{lemma}
    \begin{proof}
    We proceed with a heavy-light approach.
    Assume there exists an $O(n^{\lceil k/2 \rceil -\epsilon})$ algorithm for the $ \#k$-SUM problem.
    Set $c:=(k-1)(\lceil k/2 \rceil)$.
    Firstly, we count the number of solutions $(a_1,\dots,a_k) \in A_1 \times \dots \times A_k$, where $n_{a_1},\dots, n_{a_k}\leq n^{ \epsilon/c}$ using  Lemma \ref{reduction_t}. This takes time
    \begin{align*}
	    \Tilde{O}\left( (n \cdot (n^{ \epsilon/c})^{k-1})^{\lceil k/2 \rceil -\epsilon} \right)&=\Tilde{O} \left(\left(n^{1+\frac{\epsilon}{\lceil k/2 \rceil }} \right )^{ \lceil k/2 \rceil -\epsilon}\right) \\
                                                              &= \Tilde{O} \left(n^{ {\lceil k/2 \rceil -\epsilon} + \epsilon  -\frac{ \epsilon ^2}{\lceil k/2 \rceil } }  \right) \\
                                                              &= O \left(n^ {\lceil k/2 \rceil -\epsilon'} \right),
    \end{align*}
    where $\epsilon'>0$.
    It remains to calculate the number of witnesses $(a_1,\dots,a_k)$, where for at least one $i \in \{1,\dots,k\}$, we have high multiplicity, meaning $n_{a_i}>n^{\epsilon/c}$ holds.
    Consider the case that $a_1 \in A_1$ is a high-multiplicity number (the case where $a_i \in A_i$ with $i\neq 1$  is a high-multiplicity number is analogous).
    For each high-multiplicity number $a_1$ in $A_1$ we do the following.
    Solve the $(k-1)$-SUM instance with sets $A_2,\dots ,A_k$ and target $t-a_1$.
    There are at most $n^{1-(\epsilon/c)}$ many high-multiplicity numbers in $A_1$, and solving the $(k-1)$-SUM instance takes time $O(n^{(k-1)/2})$, since $k$ is odd.
    We get a total runtime of 
    \begin{align*}
      n^{1-\frac{\epsilon}{c}} \cdot \Tilde{O}(n^{(k-1)/2})&=\Tilde{O}(n^{1-(\epsilon/c) +(k-1)/2})\\
                    &=\Tilde{O}(n^{(k+1)/2 -(\epsilon/c)})\\
                    &=O(n^{\lceil k/2 \rceil -\epsilon''}),
    \end{align*}
    where $\epsilon''>0$, which concludes the proof.
    \end{proof}

    \countComplete*

By combining the subquadratic equivalence between $3$-SUM and $\#3$-SUM due to Chan et al.~\cite{DBLP:journals/corr/abs-2303-14572} and the above theorem, we obtain the following corollary.
\threecountcomplete*

The above proof can also be adapted for the special case $k=3$ to count for each $a_1 \in A_1$ the number of witnesses involving $a_1$, by plugging in the appropriate All-ints versions; see the full version of the paper for details.
Together with the equivalence between $\#$All-ints $3$-SUM and $3$-SUM of Chan et al. \cite{DBLP:journals/corr/abs-2303-14572}, we get
\begin{corollary}
For all problems $P$ in $\mathsf{FOP}_{\mathbb{Z}}(\exists^3)$, we are able to count for each $a_1 \in A_1$ the number of witnesses involving $a_1$ in randomized time $O(n^{2-\epsilon})$ for an $\epsilon>0$,
if $3$-SUM can be solved in randomized time $O(n^{2-\epsilon'})$ for an $\epsilon'>0$.
\label{count-all-ints}
\end{corollary}

\section{Completeness Theorems for General Quantifier Structures}\label{sec:GeneralQuantifier}
As Theorem \ref{existential_complete} establishes $3$-SUM as the complete problem for the class $\mathsf{FOP}_{\mathbb{Z}}(\exists \exists \exists)$,
we would like to similarly explore complete problems for other quantifier structures. All proofs in this section are deferred to the full version.
Let us recall our main geometric problems.
    \begin{definition}[Verification of $d$-dimensional Pareto Sum]
        Given sets $A,B,C \subseteq \mathbb{Z}^d$. Does the set $C$ dominate $A+B$, that is does for all $a \in A ,b \in B$ exist a $c \in C$, with $c \geq a+b$ ?
      \end{definition}
It is easy to see that Verification of $d$-dimensional Pareto Sum is in $\mathsf{FOP}_{\mathbb{Z}}(\forall \forall \exists)$.

\begin{definition}[Hausdorff Distance under $n$ Translations]
  Given sets $A,B,C \subseteq \mathbb{Z}^d$ with at most $n$ elements and a $\gamma\in \mathbb{N}$, the Hausdorff distance under $n$ Translations problem asks whether the following holds:
   \[\delta_{\overrightarrow{H}}^{T(A)}(B,C) \coloneqq \min_{\tau \in A} \delta_{\overrightarrow{H}}(B,C+\{\tau\}) = \min_{\tau \in A} \max_{b\in B} \min_{c\in C} \|b-(c+\tau)\|_{\infty}\leq \gamma.\]
  \end{definition}

We show the following result firstly, which allows us to assume without loss of generality a certain normal form.
\begin{lemma}
A general $\FOPZ(Q_1Q_2\exists)$ formula, with input set $A_1 \subseteq \mathbb{Z}^{d_1}, A_2 \subseteq \mathbb{Z}^{d_2},A_3 \subseteq \mathbb{Z}^{d_3} $, where $|A_1|=|A_2|=|A_3|=n$,
can be reduced to the $\FOPZ(Q_1Q_2\exists)$ formula
$$Q_1 a_1'\in A_1' Q_2 a_2' \in A_2' \exists a_3' \in A_3':a_1'+a_2' \leq a_3'$$ in time $O(n)$, where
$|A_1'|=|A_2'|=n$ and $|A_3'|=O(n)$. 
\label{normalform}
\end{lemma}
The above lemma immediately gives us complete syntactic problems for our classes.
It remains to establish connections between the different quantifier structure classes, and explore natural variants of the syntactic problems.

The syntactic complete problem for the class $\mathsf{FOP}_{\mathbb{Z}}(\exists \forall \exists)$ turns out to be equivalent to Hausdorff Distance under $n$ Translations. We obtain:

\begin{restatable}[Hausdorff Distance under $n$ Translations is complete for $\mathsf{FOP}_{\mathbb{Z}}(\exists \forall \exists)$]{lemma}{hausdorffcompl}
	There is a function $\epsilon(d)>0$ such that Hausdorff Distance under $n$ Translations can be solved in time $O(n^{2-\epsilon(d)})$ if and only if all problems $P$ in 
	$\mathsf{FOP}_{\mathbb{Z}}(\exists \forall \exists)$ can be solved in time $O(n^{2-\epsilon_P})$ for an $\epsilon_P>0$.
	\label{Hausdorff-Completeness}
\end{restatable}
Similarly, the Verification of Pareto Sum problem is complete for the class $\mathsf{FOP}_{\mathbb{Z}}(\forall \forall \exists)$.
\begin{restatable}[Verification of Pareto Sum is complete for $\mathsf{FOP}_{\mathbb{Z}}(\forall \forall \exists)$]{lemma}{verifcompl}
	There is a function $\epsilon(d)>0$ such that Verification of Pareto Sum can be solved in time $O(n^{2-\epsilon(d)})$ if and only if all problems $P$ in 
	$\mathsf{FOP}_{\mathbb{Z}}(\forall \forall \exists)$ can be solved in time $O(n^{2-\epsilon_P})$ for an $\epsilon_P>0$.
	\label{Verif-Completeness}
  \end{restatable}

\subsection{$\mathsf{FOP}_{\mathbb{Z}}(\forall \exists \exists) \to \mathsf{FOP}_{\mathbb{Z}}(\exists \exists \exists)$ }

We continue with handling the class $\mathsf{FOP}_{\mathbb{Z}}(\forall \exists \exists)$.
By simply making use of Corollary \ref{count-all-ints}, one can easily prove that $3$-SUM is hard for the class $\mathsf{FOP}_{\mathbb{Z}}(\forall \exists \exists)$.
We can also show a deterministic proof, as Corollary \ref{count-all-ints} makes use of 
the subquadratic equivalence between $3$-SUM and $\#$All-ints $3$-SUM, which relies on randomization techniques.

\allintshard*

%We remark a small observation, which will be important to us.
%\begin{observation}
 % The above proof can be generalized so that we can count the number of witnesses for each $a \in A$,
 % by making a call to $\#$All ints $3$-SUM instead of All ints $3$-SUM. This holds due to the preservation of the witnesses.
 % (Does this need an equivalence from all ints $3$-SUM auf sets und all ints $3$-SUM auf multisets ?)
%\end{observation}

\subsection{$\mathsf{FOP}_{\mathbb{Z}}(\exists \exists \exists) \to \mathsf{FOP}_{\mathbb{Z}}(\forall \forall \exists)$}
We explore the connection between the problem Additive Sumset Approximation, which is a member of the class 
$\mathsf{FOP}_{\mathbb{Z}}(\forall \forall \exists)$, and the $3$-SUM problem.
The following theorem will play a key role to enable the discovery of the relationship between $3$-SUM and other quantifier structures.
\sumsetapproxchar*
The proof is deferred to the full version of the paper.
\subsection{Completeness results for the class $\mathsf{FOP}_{\mathbb{Z}}^k$}
We turn to combining the above insights to establish (a pair of) complete problems for the class $\mathsf{FOP}_{\mathbb{Z}}$. The proofs in this section are deferred to the full version of the paper.

\verifhard*

\hunthard*
We finally obtain  our completeness theorem for the whole class $\mathsf{FOP}_{\mathbb{Z}}^k$.
\completenessclass*

Essentially, these two problems capture the complexity of the class $\mathsf{FOP}_{\mathbb{Z}}^3$ and can be seen as the most important problems in $\mathsf{FOP}_{\mathbb{Z}}^{k}$. 
%\begin{definition}[All-ints $\exists \forall$ domination]
%Given $A,B,C \subseteq \mathbb{Z}^d$ determine for each $a \in A$ whether there exists $b \in B$ such 
%that for all $c \in C$ the inequality $a+b \leq c$ holds.
%\end{definition}

%\dominancehard*
%\begin{proof}
%We bruteforce the first $k-3$ quantifiers.
%By application of Lemma \ref{normalform}, it remains to solve a formula $\phi:=Q_1 a_1 \in A_1 Q_2 a_2 \in A_2 Q_3 a_3 \in A_3: a_1 + a_2 \leq a_3$.
%By a possible negation, we are able to achieve one of the following four cases.
%\begin{enumerate}
%\item If $Q_1=Q_2\forall$ and $Q_3=\exists$, then assuming we can solve All-ints $\forall \exists$ domination in time $O(n^{2-\epsilon_d})$,
%we can conclude by checking for each $a_1$ whether the rest of the formula is satisfied .
%\item If $Q_1=Q_3\exists$ and $Q_2=\forall$, then assuming we can solve All-ints $\forall \exists$ domination in time $O(n^{2-\epsilon_d})$,
%we can conclude by checking whether for atleast one $a_1 \in A_1$ the rest of the formula is satisfied.
%\item If $Q_1=Q_2=Q_3=\exists$ or $Q_1=\forall$ and $Q_2=Q_3=\exists$, we reduce $\phi$ to a $3$-SUM instance firstly, and then to an instance of the additive aproximation problem.
%Thus, we again end up in Case 1.
%\end{enumerate}
%\end{proof}

\section{The $3$-SUM problem is complete for $\mathsf{FOP}_{\mathbb{Z}}$ formulas with Inequality Dimension at most 3}\label{sec:IneqDimension}

In this section, we show that $3$-SUM problem captures an interesting subclass of $\FOPZ$ formulas with arbitrary quantifier structure, namely the formulas of sufficiently small \emph{inequality dimension}.  Let us recall the notion of inequality dimension.
\begin{definition}[Inequality Dimension of a Formula]
  Let $\phi = Q_1 x_1\in A_1, \dots, Q_k x_k\in A_k: \psi$ be a $\FOPZ$ formula with $A_i\subseteq \mathbb{Z}^{d_i}$.
  
  The \emph{inequality dimension} of $\phi$ is the smallest number $s$ such that there exists a Boolean function $\psi' :\{0,1\}^s \to \{0,1\}$ and (strict or non-strict)
  linear inequalities $L_1, \dots, L_s$ in the variables $\{x_i[j] : i\in \{1,\dots,k\} ,j\in \{1,\dots,d_i\} \}$ and 
  the free variables such that $\psi(x_1,\dots, x_k)$ is equivalent to $\psi'(L_1,\dots, L_s)$.  
\end{definition}
In the following, we look at the class of problems $\mathsf{FOP}_{\mathbb{Z}}^k$ with the restriction of 
inequality dimension at most $3$.
We use the following naming convention for boxes.
\begin{definition}
  A $d$-box in $\mathbb{R}^d$ is the cartesian product of $d$ proper intervals $s_1 \times \dots \times s_d$,
  where $s_i$ is an open, closed or half-open interval. We call a cartesian product of only closed intervals a closed box and
  a cartesian product of only open intervals an open box.
  \end{definition}
  Given a set $R$ of $n$ closed boxes (represented as $2d$ integer coordinates), and $d$-dimensional points $a \in A ,b \in B$, we can express in $\FOPZ(\exists \exists \exists)$ whether $a+b$ lies in one of the boxes as follows:
  $$ \exists a \in A \exists b \in B \exists r \in R: \bigwedge_{i=1}^{d} r[i] \leq a[i]+b[i] \land a[i]+b[i]  \leq r[d+i]. $$ 
  In fact, we are not limited to closed boxes, if a box is open or half open in a dimension, one can adjust the inequalities in this dimension appropriately.

In order to prove our main theorem in this section, we need to partition the union of $n$ unit cubes in $\mathbb{R}^3$ into pairwise interior- and exterior-disjoint boxes.
While Chew et al.~\cite{DBLP:journals/dcg/ChewDEK99} studied such a decomposition of unit cubes  with the requirement of only interior-disjoint boxes, we need an extension of their result to guarantee disjoint exteriors.

\begin{lemma}[Disjoint decomposition of the union of cubes in $\mathbb{R}^3$]
Let $\mathcal{C}$ be a set of $n$ axis-aligned congruent cubes in $\mathbb{R}^3$. The union of these cubes,
can be decomposed into $O(n)$ boxes whose interiors and exteriors are disjoint in time $O(n \log^2 n)$.
\label{Cubes_in_space}
\end{lemma}
The proof is deferred to the full version.
\begin{theorem}
  There is an algorithm deciding $3$-SUM in randomized time $O(n^{2-\epsilon})$ for an $\epsilon>0$ if and only if
  for each problem $P$ in the classes $\mathsf{FOP}_{\mathbb{Z}}(\forall \forall \exists )$ and $\mathsf{FOP}_{\mathbb{Z}}(\exists \forall \exists )$ of inequality dimension at most $3$ there exists some $\epsilon'>0$ such that we can solve $P$ in randomized time $O(n^{2-\epsilon'})$.
  \label{ineq3}
\end{theorem}
\begin{proof}
  For the first direction due to Theorem \ref{sumsetapproxTHM}, we can reduce $3$-SUM to an instance of Additive Sumset Approximation,
  $$ \forall a \in A \forall b \in B \exists c \in C: c \leq a+b \land a+b \leq c+t,$$
  which has inequality dimension 2. Let us continue with the other direction.
  Let $\phi:=Q_1 a \in A \forall b \in B \exists c \in C: \varphi$, where $Q_1 \in \{\exists,\forall \}$ and $\varphi$ is a quantifier free linear arithmetic formula with inequality dimension $3$.
  Let $L_1:=\alpha_{1}^T a + \beta_{1}^T b \leq \gamma_{1}^T c +S_{1}$,
   $L_2:= \alpha_{2}^T a + \beta_{2}^T b \leq \gamma_{2}^T c +S_{2}$ and
   $L_3:= \alpha_{3}^T a + \beta_{3}^T b \leq \gamma_{3}^T c +S_{3}$ after replacing the free variables.
  Assume that the formula $\varphi$ is given in DNF, thus each co-clause
  has at most $3$ atoms, chosen from $L_1,L_2,L_3$ and their negations. 
  Let 
  \begin{align*}
   A':= \left \{ \left( \begin{array}{cc}
    \alpha_{1}^T a \\
    \alpha_{2}^T a \\
    \alpha_{3}^T a
    \end{array}  \right): a \in A \right \},
    B':=\left \{ \left( \begin{array}{cc}
      \beta_{1}^T b \\
      \beta_{2}^T b \\
      \beta_{3}^T b 
      \end{array}  \right):b \in B \right \},
      C':=\left \{ \left( \begin{array}{cc}
        \gamma_{1}^T c +S_{1} \\
        \gamma_{2}^T c +S_{2}\\
        \gamma_{3}^T c+ S_{3}
        \end{array}  \right):c \in C \right \}
  \end{align*}
  Thus each co-clause consists of conjunctions of a subset of the following set
  \begin{align*}
    \{ & a'[0]+b'[0] \leq c'[0] ,a'[0]+b'[0] \geq c'[0]+1, 
                    a'[1]+b'[1] \leq c'[1], \\
                   & a'[1]+b'[1] \geq  c'[1]+1, 
                    a'[2]+b'[2] \leq c'[2], a'[2] +b'[2] \geq c'[2] +1
    \}.
  \end{align*}
  Let the co-clauses of $\varphi$ be $V_1, \dots ,V_h$. Thus, we aim to decide a formula of the form:
  \begin{equation}\label{eq:Dnfcurr}
  Q_1 a' \in A' \forall b' \in B' \exists c' \in C': \bigvee_{i=1}^{h} V_{i} 
  \end{equation} 
  For each co-clause $V_i$, $i \in \{1,\dots,h\} $ it holds that $V_i$ is of the form
  $$ \bigwedge_{k \in V_i^K} L_k  \land \bigwedge_{j \in V_i^J} \lnot L_j,$$
  for some $V_i^J,V_i^K \subseteq \{1,2,3\}$ and $V_i^J\cap V_i^K =\emptyset$. 

  Let us consider for each fixed $c' \in C'$ the following possibly empty orthant in $\mathbb{R}^3$.
  $$\mathcal{S}(V_{i},c'):=\{x \in \mathbb{R}^3 : \bigwedge_{k \in V_i^K}x[k] \leq c'[k] \land \bigwedge_{j \in V_i^J}x[j] \geq c'[j]+1   \}.$$

  By construction, it is immediate that for a fixed $c'$  and $(a',b') \in A' \times B'$ that $(a',b',c')$ fulfill the co-clause $V_i $ if and only if $a'+b' \in \mathcal{S}(V_{i},c')$.
  Thus, equivalently to \eqref{eq:Dnfcurr}, we ask
	\[ Q_1 a' \in A' \forall b' \in B' \exists c' \in C': \bigvee_{i=1}^{h} \left( a'+b' \in S(V_i,c') \right).\]
  Having a closer look, $\bigvee_{i=1}^{h} \left( a'+b' \in S(V_i,c') \right)$ is true if and only if $a'+b'$ lies in one of the orthants $S(V_i,c')$.

	We argue that we may represent the orthant $ \mathcal{S}(V_{i},c')$ as an appropriately chosen cube in $\mathbb{R}^3$. To this end, let $M:=2 \cdot \max \{\|a \|_1+ \|b \|_1+ \|c \|_1: a' \in A', b' \in B', c' \in C'\}$ be a sufficiently large number. We can interpret $\mathcal{S}(V_{i},c')$ as a cube of the type $\mathcal{C}_{i,c'} = [m_0,m'_0] \times [m_1,m_1'] \times [m_2,m'_2]$,
  where for $u \in \{0,1,2\}$, we define:
  \begin{align*}
  m_u:=\begin{cases} 
    -M & u \not \in V_i^K , u \not \in V_i^J,\\
    -2M+c[u] & u \in V_i^K, \\
    c[u]+1 & u \in V_i^J,
 \end{cases} \quad   m_{u}' := \begin{cases} 
  M & u \not \in V_i^K , u \not \in V_i^J,\\
  c[u] & u \in V_i^K, \\
  2M+c[u]+1 & u \in V_i^J. 
\end{cases}
\end{align*}
The cubes are axis-aligned and have side length $2M$. Due to the large size of the cube we get for fixed $c' \in C'$ that
$a'+b' \in \mathcal{S}(V_i,c')$ if and only if $a'+b'$ lies inside the cube $\mathcal{C}_{i,c'}$.

By Lemma \ref{Cubes_in_space}, we can decompose the collection of cubes $\mathcal{C}_{i,c'}$ for $i \in \{1,\dots,H \}, c' \in C'$ into $l = O(n)$ disjoint boxes $\mathcal{R}:=\{R_1, \dots, R_l\}$ in time $O( n \log^2 n)$.
Let us now go through a case distinction based on the first quantifier. 
  \begin{itemize}
  \item If $Q_1=\forall$, equivalent to $\phi$ we ask 
  $$\forall a' \in A' \forall b' \in B' \exists i \in \{1,\dots,l\}: a'+b' \text{ lies in } R_i.$$
  
		  By replacing each $i \in \{1,\dots,l\}$ by a 6-tuple denoting the dimensions of the box $R_i$, we can reduce counting the number of $(a',b',R_i)$ with $a'+b' \in R_i$ to 3-SUM using Corollary~\ref{Count3COMP}.
		  Due to the disjointness of the boxes $R_i$, we know that no $(a',b')$ can be in different boxes $R_i, R_{i'}$ with $i\ne i'$.

  Thus, we can decide our original question by checking whether the number of such witnesses equals $|A'| \cdot |B'|$, concluding the fine-grained reduction to 3-SUM.

  \item Assume now that $Q_1=\exists$. Thus, equivalently to $\phi$, we ask.
  $$\exists a' \in A' \forall b' \in B' \exists i \in \{1,\dots,l\}: a'+b' \text{ lies in } R_i.$$

		  We can now make use of Corollary~\ref{count-all-ints}. Count for each $a' \in A'$ the number of \emph{witnesses} $(a',b',R_i)$ with $a'+b'\in R'$. We claim that it remains to check
whether there is some $a'$ that is involved in $|B'|$ witnesses. 
		  To see this, note that due to the disjointness of the $R_{i}$'s, for any $a'\in A'$ we have that the number of $(b',R_i)$ with $a'+b'\in R_i$ is equal to the number of $b'$ such that there exists $R_i$ with $a'+b'\in R_i$. Again, the desired reduction to 3-SUM follows. \qedhere
  \end{itemize}  
\end{proof}

We remark that, by \cite{DBLP:journals/dcg/BoissonnatSTY98}, we know that the complexity of the union of orthants in $\mathbb{R}^d$ has worst case complexity $O(n^{\lfloor d/2 \rfloor})$.
Thus, the above proof does not seem directly generalizable to inequality dimensions larger than 3.
We can extend Theorem \ref{ineq3} to $k$-quantifiers by the following theorem.

\threesumineq*

  The above theorem gives us immediate reductions to $3$-SUM for many seemingly unrelated problems of different 
  quantifier structures and semantics.

  For instance, as a direct application of the above theorem we can conclude the equivalence of the 
  Additive Sumset Approximation problem to $3$-SUM, together with Theorem \ref{sumsetapproxTHM}.

  \begin{lemma}[Additive Sumset Approximation $\leq_2$ $3$-SUM]
    If the $3$-SUM problem can be solved in randomized time $O(n^{2-\epsilon})$ for an $\epsilon>0$
    then Additive Sumset Approximation problem can be solved in randomized time 
    $O(n^{2-\epsilon'})$ for an $\epsilon'>0$.
    \end{lemma}

\section{Application:  A lower bound on the computation of Pareto Sums} \label{sec:applications} \label{sec:ParetoSum}
In the following, we explore how the $3$-SUM hardness of Verification of Pareto Sum
translates to a hardness result for the problem of computing Pareto Sums.
Let us first justify the naming of the Verification of Pareto Sum problem,
by showing it to be subquadratic equivalent to the more natural extended version of Verification of Pareto Sum.
Throughout this section, we consider dimensions $d\ge 2$.
\begin{definition}[Verification of Pareto Sum (Extended version)]
	Given sets $A,B,C \subseteq \mathbb{Z}^d$, do 
	the following properties hold simultaneously:
	\begin{itemize}
	\item (Inclusion): $C\subseteq A+B$,
	\item (Dominance): $C$ dominates $A+B$. More formally, for every $a \in A, b \in B$ there exists $c \in C$ with $c\geq a+b$.
	\item (Minimality): There are no $c,c' \in C$ with $c\neq c'$ and $c \leq c'$. 
	\end{itemize}
\end{definition}

We make use of the following lemma and its construction for the results in this section.

\begin{lemma}
	Given sets $A,B,C \subseteq \mathbb{Z}^d$ of size at most $n$, one can construct sets $\Tilde{A},\Tilde{B},\Tilde{C} \subseteq \mathbb{Z}^d$ of size $\Theta(n)$ 
	in time $\Tilde{O}(n)$ such that (1) $\Tilde{A}, \Tilde{B}, \Tilde{C}$ always satisfy the minimality and inclusion condition and (2)
	$\Tilde{A},\Tilde{B},\Tilde{C}$ fulfill the dominance condition if and only if $A,B,C$ fulfill 
	the dominance condition.  
	\label{annoying_construction}
	\end{lemma}
	Due to space constraints, the proof had to be deferred to the full version. Using this construction, it is not difficult to obtain the following equivalence.

	\begin{lemma}
		There is an $O(n^{2-\epsilon})$ time algorithm for an $\epsilon>0$ for Verification of Pareto Sum (Extended Version)
		if and only if there is an $O(n^{2-\epsilon'})$ time algorithm for an $\epsilon'>0$ for Verification of Pareto Sum.
		\end{lemma}
	The proof can be found in the full version.
Thus, for subquadratic reductions, we can restrict ourselves to the Verification of Pareto Sum problem, which essentially
only checks the dominance condition.

Let us now consider the natural problem of computing the Pareto Sum.
\begin{definition}[Pareto Sum]
	Given sets $A,B \subseteq \mathbb{Z}^d$, compute a set $C \subseteq \mathbb{Z}$,
	such that $A,B,C$ satisfy the Inclusion, Dominance and Minimality condition.
\end{definition}

In the following, we argue why the lower bounds to Verification of Pareto Sum translate to lower bounds 
to Computation of the Pareto Sum. Formally, we prove:

\begin{lemma}
If there is an algorithm to compute the Pareto Sum $C$ of sets $A,B \subseteq \mathbb{Z}^d$ in time 
	$O(n^{2-\epsilon})$ for an $\epsilon>0$ even when $C=\Theta(n)$,
then one can also decide Verification of Pareto Sum of sets $A,B,C$ in time $O(n^{2-\epsilon'})$ for an
$\epsilon'>0$. 
\label{Verif to Comp}
\end{lemma}

We conclude this section with our resulting hardness results for computing Pareto Sums.
\lowerboundPS*

\section{Future Work}\label{sec::ende}
While we exhibit a pair of problems that is complete for the class $\mathsf{FOP}_{\mathbb{Z}}$,
one could still ask whether there is a subquadratic reduction from Hausdorff distance under $n$ Translations to Verification of Pareto Sum.
As a result there would be a single complete problem (or rather the canonical multidimensional family of a single geometric problem) for $\mathsf{FOP}_{\mathbb{Z}}$.
\begin{center}
    \emph{ Is Verification of Pareto Sum complete for the class $\mathsf{FOP}_{\mathbb{Z}}$?}
    \end{center}

Interestingly, previous completeness theorems \cite{GaoIKW19} were able to establish a problem of quantifier structure $\forall \forall \exists$ (the Orthogonal Vectors problem) as complete
by making use of a technique in~\cite{DBLP:conf/focs/WilliamsW10} that was originally used to show subcubic equivalence
between All-Pairs Negative Triangle and Negative Triangle.
However, a major problem we encounter is that while the third quantifier in the Orthogonal Vectors problem ranges over a sparse (intuitively: subpolynomially sized) domain (i.e., the dimensions of the vectors),
the third quantifier in Pareto Sum Verification ranges over a linearly sized domain (i.e., the set $C$).

Finally, we ask if our $3$-SUM completeness result for arbitrary quantifier structures can be  improved upon.
\begin{center}
\emph{Can we establish a $d>3$ such that $3$-SUM is complete for 
$\mathsf{FOP}_{\mathbb{Z}}$ formulas of inequality dimension at most $d$?}
\end{center}

%%
%% Bibliography
%%

%% Please use bibtex, 

\bibliography{descrip3sum}
\appendix

\end{document}